\documentclass[letterpaper,11pt]{article}
\newif\ifdraft
\draftfalse
\usepackage[utf8]{inputenc}
\usepackage{microtype}
\usepackage[letterpaper,margin=1in]{geometry}
\usepackage[numbers,sort&compress]{natbib}
\usepackage{xcolor,xspace}
\usepackage{amsmath}
\usepackage{amssymb}
\usepackage{amsthm}
\usepackage{enumitem}
\usepackage{textgreek}
\usepackage{graphicx}
\usepackage{framed}

\usepackage[small]{caption}
\usepackage{booktabs}
\usepackage[unicode]{hyperref}
\usepackage{cleveref}
\usepackage{doi}
\usepackage{thmtools,thm-restate}
\usepackage[noend]{algorithm2e}
\usepackage[noend]{algpseudocode}
\usepackage{xspace}
\usepackage{nicefrac}

\theoremstyle{plain}
\newtheorem{theorem}{Theorem}[section]
\newtheorem{lemma}[theorem]{Lemma}

\newtheorem{corollary}[theorem]{Corollary}
\newtheorem{claim}[theorem]{Claim}

\newtheorem*{conjecture*}{Conjecture}
\newtheorem{observation}[theorem]{Observation}

\theoremstyle{definition}

\theoremstyle{remark}
\newtheorem{remark}[theorem]{Remark}
\newtheorem*{remark*}{Remark}

\newcommand{\CONGEST}{\ensuremath{\mathsf{CONGEST}}\xspace}
\newcommand{\LOCAL}{\ensuremath{\mathsf{LOCAL}}\xspace}

\newcommand{\eps}{\varepsilon}
\newcommand{\poly}{\operatorname{\text{{\rm poly}}}}

\newcommand{\logstar}[1]{\log^{*} #1}

\DeclareMathOperator{\polylog}{\poly\log}

\newcommand{\NN}{\mathbb{N}}

\newcommand{\FF}{\mathbb{F}}

\renewcommand{\phi}{\varphi}

\newenvironment{myabstract}
{\list{}{\listparindent 1.5em%
        \itemindent    \listparindent
        \leftmargin    1cm
        \rightmargin   1cm
        \parsep        0pt}%
    \item\relax}
{\endlist}

\newenvironment{mycover}
{\list{}{\listparindent 0pt
        \itemindent    \listparindent
        \leftmargin    1cm
        \rightmargin   1cm
        \parsep        0pt}%
    \raggedright
    \item\relax}
{\endlist}

\newcommand{\myaff}[1]{\,$\cdot$\, {\small #1}\par\smallskip}

\hypersetup{
    colorlinks=true,
    linkcolor=black,
    citecolor=black,
    filecolor=black,
    urlcolor=[HTML]{0088CC},
    pdftitle={Distributed Graph Coloring Made Easy},
    pdfauthor={Yannic Maus}
}

\begin{document}

\begin{mycover}
    {\huge\bfseries Distributed Graph Coloring Made Easy \par}
    \bigskip
    \bigskip

	\textbf{Yannic Maus\footnote{This research was mostly conducted while the author was employed by the Technion in Israel.}}
	\myaff{TU Graz, Austria}
\end{mycover}
\medskip

\begin{myabstract}
\noindent\textbf{Abstract.}
	In this paper we present a deterministic \CONGEST algorithm  to compute an $O(k\Delta)$-vertex coloring in $O(\Delta/k)+\logstar n$ rounds, where $\Delta$ is the maximum degree of the network graph and $1\leq k\leq O(\Delta)$ can be freely chosen. The algorithm is extremely simple: Each node locally computes a sequence of colors and then it \emph{tries colors} from the sequence in batches of size $k$.  	
	Our algorithm subsumes many important results in the history of distributed graph coloring as special cases, including Linial's color reduction [Linial, FOCS'87], the celebrated locally iterative algorithm from [Barenboim, Elkin, Goldenberg, PODC'18], and various algorithms to compute defective and arbdefective colorings. Our algorithm can smoothly scale between these and also simplifies the state of the art $(\Delta+1)$-coloring algorithm. 	At the cost of losing the full algorithm's simplicity  we also provide a $O(k\Delta)$-coloring algorithm in $O(\sqrt{\Delta/k})+\logstar n$ rounds. We also provide improved deterministic algorithms for ruling sets, and, additionally, we provide a tight characterization for one-round color reduction algorithms.
\end{myabstract}

\section{Introduction}
\label{sec:intro}
In the \emph{$C$-vertex coloring} problem the objective is to assign each vertex of an $n$-node graph $G=(V,E)$ one of $C$ colors such that adjacent vertices get different colors.  In the distributed setting, graph coloring  is considered to be one of the core \emph{symmetry breaking problems} with a huge amount of published work and even a whole book almost exclusively covering the problem  \cite{barenboimelkin_book}. In this setting, the usual goal is to compute a $C$-coloring with $\Delta+1\leq C \leq O(\Delta^2)$, where $\Delta$ is the maximum degree of the graph. The bound of $\Delta+1$ stems from the fact that any graph can be colored with the respective number of colors and this can even be done with a simple sequential greedy algorithm. The bound of $O(\Delta^2)$ colors stems from an algorithm in Linial's seminal paper in which he introduces one of the core models for distributed graph algorithms, i.e., the \LOCAL model \cite{linial92}. In this model,  the graph abstracts a communication network in which the nodes communicate through the edges in synchronous rounds and at the end of the computation each node needs to output its own part of the solution, e.g., its own color. The complexity measure is the number of synchronous rounds. Linial gave an extremely fast $O(\Delta^2)$-coloring algorithm that only uses $O(\logstar n)$ rounds. Further, he showed that $\Omega(\logstar n)$ rounds are needed to color rings ($\Delta=2$) with $O(1)=\Delta^{O(1)}$ colors.
Due to the lower bound, a vast amount  of published papers, e.g., \cite{SzegedyV93,KuhnW06,BarenboimEK14,FHK,Barenboim16,BEG18,MT20}, study the setting that occurs after applying Linial's coloring algorithm, i.e., they ask: \emph{Given an $O(\Delta^2)$-coloring, how fast can one reduce the number of colors where the runtime of the algorithm can only depend on $\Delta$?} The similar question of finding a fast algorithm with complexity $f(\Delta)+\logstar n$ is also sometimes referred to as determining the \emph{truly local complexity} of a problem \cite{MT20}. 

In the current paper, we devise several results to advance the understanding of this setting. First, we provide a \emph{simple} deterministic algorithm that scales between the two extremes of $\Delta+1$ (or rather $O(\Delta)$ colors) and $O(\Delta^2)$ colors. In particular, for a parameter  $k\geq 1$ of the user's choice the algorithm computes a $O(k\Delta)$ coloring with complexity $f(\Delta)=O(\Delta/k)$. In the algorithm, each node $v$ uses its input color, e.g., provided by Linial's algorithm, to (locally) compute a sequence $p_v(0),p_v(1), \ldots $ of colors. Then, node $v$ \emph{tries} to get colored with one of the first $k$ colors in its sequence by sending these trials to its neighbors and receiving their trials. If $v$ tries a color $c$ that is not \emph{conflicting} with the neighbors' trials node $v$ gets permanently colored with color $c$, otherwise $v$ continues to the next round in which it tries the next $k$ colors in its sequence, and so on.  

Second, we show that this simple \emph{mother algorithm} either immediately yields the core steps of the aforementioned papers, e.g., the algorithms in \cite{linial92,BEG18}, or can be mildly adapted to obtain crucial subroutines developed or used in \cite{BarenboimEK14,Kuhn2009WeakColoring,BE09,FHK,Barenboim16,BEG18,MT20}. E.g., a mild adaptation yields \emph{$d$-defective colorings} which were the crucial ingredient in \cite{BarenboimEK14, Kuhn2009WeakColoring,BE09}, or \emph{arbdefective colorings} which were crucial in \cite{FHK, Barenboim16, BEG18, MT20}. A core strength of our result is that the algorithms for each of these results are very similar.  To get a feeling for this let us look at defective colorings. In a \emph{$d$-defective coloring} a node is allowed to have at most $d$ neighbors with the same color.  Besides a suitable choice of $k$ and sequences $p_v(0),p_v(1), \ldots$, the only change is in the execution of the algorithm: When deciding whether to keep a color $c$ a node tolerates up to $d$ neighbors with the same color. This algorithm does not yet do the job as the defect of a node $v$ might be larger than $d$ at the end of the execution if one or more neighbors of $v$ choose the same color as $v$ in later rounds, but it still captures the essence of the adaptations that need to be performed. 

Third, at the cost of losing the full simplicity of the algorithm, we show how to compute an $O(\Delta^{1+\eps})$-coloring in $O(\Delta^{1/2-\eps/2})$ rounds, and we improve the state of the art runtime for computing so called \emph{ruling sets}.

Fourth, we provide a full characterization of $1$-round coloring algorithms. Informally, we determine the maximum number $q_{m,\Delta}$ of colors that can be reduced by a $1$-round coloring algorithm that works for any graph with a given maximum degree $\Delta$ and $m$-input coloring.

Next, we describe why this result, combined with our $O(k\Delta)$-coloring algorithm in $O(\Delta/k)$ rounds might be of additional interest. Already \cite{SzegedyV93,KuhnW06,disc16_coloring} studied $1$-round color reduction algorithms and showed lower bounds like our fourth contribution; then in \cite{SzegedyV93} these $1$-round lower bounds on the number of colors were used to prove a heuristic $\Omega(\Delta\log \Delta)$ runtime lower bound for computing a $(\Delta+1)$-coloring. As already pointed out by \cite{SzegedyV93} the bound is \emph{heuristic} in the following sense (the following example uses the result of our paper): Given a coloring with at most $2\Delta$ colors, we can reduce exactly one color in a single round. By applying this tight bound twice,  one would wish to claim that one cannot go from a $2\Delta$ coloring to a $2\Delta-2$ coloring in two rounds. However, this claim cannot be deduced via this method, as the second application of the $1$-round lower bound assumes that the \emph{intermediate $2\Delta-1$ coloring} is \emph{worst case}. But, instead a $2$-round algorithm might not produce an intermediate coloring at all,  or it might output a very specific intermediate coloring which enables it to reduce more than one color in its second round. The heuristic lower bound in \cite{SzegedyV93} is obtained by applying $1$-round lower bounds iteratively, purposely ignoring this important subtlety. 
Since the publication of \cite{SzegedyV93} at least five different algorithms for $(\Delta+1)$-coloring that beat this lower bound were published: Given an $O(\Delta^2)$-coloring, a $(\Delta+1)$-coloring  was computed in $O(\Delta)$ rounds \cite{BE09, Kuhn2009WeakColoring, BarenboimEK14}, in $O(\Delta^{3/4})$ rounds in \cite{Barenboim16}, in $O(\sqrt{\Delta\log \Delta}\logstar \Delta)$ rounds in \cite{FHK, BEG18} and in $O(\sqrt{\Delta\log \Delta})$ rounds in \cite{MT20}. Formally, only the \emph{locally iterative}  $O(\Delta)$-round algorithm by Barenboim, Elkin and Goldenberg \cite{BEG18} beats this lower bound and we explain why. It starts with an $O(\Delta^2)$ coloring and  maintains a feasible coloring in each round. The color of a node in the next round only depends on the node's own color and the colors of its neighbors. The algorithm was celebrated as it is significantly simpler than the aforementioned faster algorithms and it breaks the aforementioned heuristic lower bound in a \emph{clean way}, due to maintaining a feasible coloring in each round. 

 While we formally do not maintain a feasible coloring in each round in our algorithms\footnote{It is straightforward to tweak the algorithms to actually achieve this at the cost of losing some of the simplicity of the algorithm, e.g., by encoding the state of a node into a proper vertex coloring.},  we provide a different insight.  
 Our \emph{tight} lower bounds for $1$-round algorithms  provide a heuristic argument that shows that one needs $\Omega(\Delta)$ rounds to go from a $O(\Delta^2)$ coloring to a $\Delta^2/5$-coloring (the $5$ is chosen somewhat arbitrarily), and for a suitable $k=\Omega(\Delta)$, our mother algorithm provides a $O(\Delta/k)=O(1)$-round algorithm to perform such a color reduction.  Thus the  heuristic lower bound that is based on the repeated application of \emph{tight} $1$-round lower bounds can be beaten significantly in the number of colors by a simple $O(1)$-round algorithm. In contrast, all the previous algorithms that provide such an insight require $\Delta^{O(1)}$ rounds. 
 This is in particular interesting, as there has been almost no progress in proving lower bounds for the $(\Delta+1)$-coloring problem since Linial's initial seminal $\Omega(\logstar n)$ lower bound. Just as we do in this paper, the only other known lower bounds in the \LOCAL model study $1$-round algorithms \cite{SzegedyV93, KuhnW06,disc16_coloring}. 
Our paper suggests that  even in constant-time coloring algorithms there are still results to be discovered.  In fact, we do not even know a lower bound on the number of colors that a $2$-round algorithm must use. Surprisingly, there is not even a lower bound, that rules out that one can go from a $O(\Delta^2)$ coloring to a $\Delta+1$ coloring in two rounds.
 Further evidence that understanding constant-time or even just $2$-round algorithms is given by \cite{MT20}. It provides a $2$-round algorithm for a \emph{list coloring} variant of Linial's color reduction. In \emph{list coloring} each node $v$ has a list $L(v)$ of colors and needs to output a color from this list. Basically, the authors show that one can compute a list-coloring in two rounds if lists are of size $\Tilde{\Omega}(\Delta^2)$---the 'equivalent' of the $O(\Delta^2)$ colors in Linial's coloring---, and interestingly for their choice of parameters the exact same problem cannot be solved in one round. Of course, just as applying a lower bound for $1$-round algorithms twice does not give a tight lower bound for two rounds it is unclear whether understanding $2$-round algorithms will yield a result for the holy grail, a tight runtime lower bound for $(\Delta+1)$-coloring.
\subsection{Our Contributions}
While we have already explained our contributions from a high level point of view we use this section to state them formally, additional related work is presented afterwards. Our results hold in the \LOCAL model and in the \CONGEST model (both lower and upper bounds). 

\paragraph{The \LOCAL and \CONGEST Model of distributed computing~\cite{linial92,peleg00}.} In both models the graph is abstracted as an $n$-node network $G=(V, E)$ with maximum degree at most $\Delta$. Communication happens in synchronous rounds. Per round, each node can send one message to each of its neighbors. At the end, each node has to know its own part of the output, e.g., its own color. In the \LOCAL model there is no bound on the message size and in the \CONGEST model messages can contain at most $O(\log n)$ bits. Usually, in both models nodes are equipped with $O(\log n)$ bit IDs and initially, nodes know their own ID or their own color in an input coloring but are unaware of the IDs of their neighbors. 

Linial's algorithm treats the unique IDs as an input coloring to compute an $O(\Delta^2)$-coloring in $O(\logstar n)$ rounds, merely in one \emph{color reduction step} he reduces an $m$-input coloring to an $O(\Delta^2\poly\log m)$-coloring,  which then serves as the input coloring for the next step. All of our algorithms do not make use of unique IDs but work in the more general setting where nodes are only equipped with some input coloring with $m$ colors.  Similarly to most previously-known results, we assume that $m$ and $\Delta$ (and sometimes additional parameters) are global knowledge. It is easiest to grasp our results when setting $m=O(\Delta^2)$, that is, one first applies Linial's algorithm. Our main technical result is the following theorem. 
\begin{restatable}{theorem}{motherAlgorithm} \label{thm:mainSimple}
	There exists a distributed deterministic algorithm that performs as follows in any undirected graph $G=(V,E)$: 
	\begin{description}
		\item[Input:] At every node $v\in V$, the algorithm takes as input an integer $m\geq 1$, a color $c_v\in [m]$ such that the colors of the vertices form an $m$-coloring of $G$,  the maximum degree $\Delta$ of $G$, and two integers $d, k$ where $0\leq d\leq \Delta-1$ and $1\leq k\leq X$ for $Z=\frac{\Delta}{(d+1)}$ and $X=4\cdot Z\cdot \lceil\log_{Z}m\rceil$.
		
		\item[Output:] At every node $v\in V$, the algorithm outputs a color in $[kX]$ such that 
		\begin{enumerate}
			\item the graph induced by each color class admits an orientation of its edges with outdegree at most $d$~,
			
			\item each color class can be partitioned into $R=\lceil X/k\rceil$ induced subgraphs $P_1,\ldots,P_R$ of degree at most $d$~.
		\end{enumerate}
	\end{description}
	This algorithm runs in $R=\lceil X/k\rceil$ rounds in the \CONGEST model. The orientation and the partition can be computed with no additional cost in the round complexity.
\end{restatable}
Note that whenever $d\neq 0$, the coloring computed by the algorithm of \Cref{thm:mainSimple} may not be  proper, i.e., neighboring vertices may output the same color. \Cref{cor:allInOne} summarizes the most important parameter settings for \Cref{thm:mainSimple}, including settings to compute proper colorings $(d=0)$. 
While the algorithm for \Cref{thm:mainSimple} is extremely simple (locally compute a permutation of the output colors, try them in  batches of size $k$ and tolerate up to $d$ conflicts), the theorem, stated in its general form, has a rather technical appearance to fit various choices of parameters at once. But, we believe that it is very approachable as soon as one considers precise choices for its parameters. E.g., we can first use Linial's algorithm (or an algorithm derived from \Cref{thm:mainSimple}) to compute a $O(\Delta^2)$-coloring in $O(\logstar n)$ rounds. If we treat this coloring as an input coloring with $m=O(\Delta^2)$ colors and if also $d=0$ or $d=\Delta^\eps$ holds for some constant $0<\eps<1$, one can replace the term $\lceil\log_{Z}m\rceil$ with a constant.  

One step of Linial's color reduction \cite{linial92} is based on a suitable so-called \emph{low-intersecting set family} $S_1,\ldots,S_m$; Linial uses the probabilistic method to show that the respective families exist. In the algorithm a node with input color $i$ simultaneously tries all colors in $S_i$, and as $S_i\cap S_j$ is \emph{small} for each neighbor's input color $j\neq i$, it is guaranteed that at least one color in $S_i$ is not tried by any neighbor.  The low-intersecting set families obtained by the probabilistic method are not strong enough to go from a \mbox{$\poly\Delta$-coloring} to an $O(\Delta^2)$-coloring. Hence Linial uses a different construction for such families, based on polynomials.\footnote{Szegedy and Vishnawatan  show how to compute a coloring with $\poly \Delta$ colors in $0.5\logstar n$ rounds, if this algorithm is followed by $O(1)$ iterations of Linial's color reduction based on polynomials, this implies an $O(\Delta^2)$-coloring algorithm in $0.5\logstar n +O(1)$ rounds \cite{SzegedyV93}.} At its core is the fundamental theorem that for two distinct polynomials $p_1$ and $p_2$ with degree $d$ over a suitable finite set $\mathbb{F}_q$, the sets $S_i=\big\{\big(x,p_i(x)\big) \mid x\in\mathbb{F}_q\big\}$, $i=1,2$ intersect in at most $d$ elements, see \cite[Example 3.2]{EFF85}. Choosing $m$  distinct polynomials yields the respective set family with $m$ sets. 
This argument is also the core of our main result; in particular Linial's $1$-round color reduction is a special  case of our more general \Cref{thm:mainSimple}. Next, we discuss various settings for the parameters in \Cref{thm:mainSimple} and explain which results it subsumes. An approachable summary is contained in \Cref{cor:allInOne}. If $d=0$,  the computed coloring is proper and point (1) and (2) in \Cref{thm:mainSimple} can be ignored. The parameter $k$ trades the number of rounds versus the number of colors. For the extreme choice of $k=X=O(\Delta)$ we obtain the aforementioned color reduction by Linial (the one build on top of polynomials), and for $k=1$ we obtain a generalization of the locally iterative algorithm of \cite{BEG18}. Other values of $k$ scale between both algorithms and provide an extremely simple way to compute $O(\Delta\cdot k)$-coloring in $O(\Delta/k)$ rounds.  While our algorithm for $k=1$ only computes an $O(\Delta)$-coloring in $O(\Delta)$ rounds, we can use an additional $O(\Delta)$ rounds in each of which we remove a single color class to transform it into a $(\Delta+1)$-coloring.

We now explain \Cref{thm:mainSimple} in the case of $d>0$. 
A \emph{$\beta$-out degree $c$-coloring} is a vertex coloring with $c$ colors together with an orientation of the edges between neighbors with the same color such that each node has at most $\beta$ outgoing edges. Note that the edges between vertices with different colors do not need to be oriented.\footnote{These colorings with a bound on the outdegree are closely related to \emph{arbdefective} colorings which were introduced in \cite{barenboimE10} and have played a significant role in the development of sublinear in $\Delta$ algorithms (more details in \cite{MT20}).}
Consider the setting of $k=1$ and $d=\beta=\Delta^{\eps}$ for a constant $0<\eps<1$. Due to the first condition of \Cref{thm:mainSimple}, we obtain a simple $\beta$-out degree $O(\Delta/\beta)$-coloring algorithm that runs in $O(\Delta/\beta)$ rounds. Further, by assigning a vertex $v$ with output color $\phi(v)$ the color tuple $(\phi(v),i)$ where $i$ is the index of the subgraph $P_i$ that it belongs to in $(2)$, we obtain a $d$-defective $O((\Delta/d)^2)$-coloring in $O(\Delta/d)$ rounds. For $k=X=O(\Delta)$, the same defective coloring can be computed in one round. This simplifies and subsumes several results in the literature. 

A $\beta$-out degree $O(\Delta/\beta)$-coloring algorithm is one of the two crucial components in the state of the art $(\Delta+1)$-coloring algorithm in \cite{MT20}. Our simpler algorithm to compute such a coloring thus also simplifies the overall algorithm, see \Cref{ssec:simplified} for details. 

The next corollary summarizes the  parameter settings in \Cref{thm:mainSimple}, that are most interesting with our current knowledge. In the future, other settings of parameters might be of interest. 
\begin{restatable}{corollary}{corAllinOne}\label{cor:allInOne}
	There are the following deterministic CONGEST algorithms that compute the stated proper colorings in the stated runtimes on any $\Delta^4$-input colored graph with maximum degree $\Delta$,  given a  globally known parameter $k\geq 1$ (in 2. and 3.):
	\begin{enumerate}
		\item $256\Delta^2$-coloring in $1$ round.  (Linial's color reduction \cite{linial92})
		\item $16\Delta\cdot k$-coloring in $O(\Delta/k)$ rounds. (subsumes results in \cite{BE09,BarenboimEK14,Kuhn2009WeakColoring})
		\item $\Delta^2$-coloring in $O(1)$ rounds. 
		
	\end{enumerate}
There are the following deterministic CONGEST algorithms that compute the stated improper colorings in the stated runtimes on any $\Delta^4$-input colored graph, given a globally known parameter $\beta=\Delta^{\eps}$ (in 4.) or  $d=\Delta^{\eps}$ (in 5. and 6.) for a constant $0<\eps<1$:
	\begin{enumerate}[resume]
		\item $\beta$-out degree $O(\Delta/\beta)$-coloring in $O(\Delta/\beta)$ rounds (subsumes a result in \cite{BEG18})
		\item $d$-defective $O((\Delta/d)^2)$-coloring in $1$ round (subsumes a result in  \cite{Kuhn2009WeakColoring,BarenboimEK14})
		\item $d$-defective $O((\Delta/d)^2)$ coloring in $O(\Delta/d)$ rounds (subsumes some results in \cite{BE09,BarenboimEK14})
	\end{enumerate}
\end{restatable}
 The required $\Delta^4$-input coloring can be computed with Linial's algorithm for a sufficiently large constant $\Delta$. The precise choice of the constants in the $O$-notation in the defective coloring in \Cref{cor:allInOne} depends linearly on the exponent $\eps$.
The algorithm of (6) is clearly inferior to the one in (5),  as it computes a $d$-defective coloring with the same number of colors but is slower. We merely state (6) for its proof (see \Cref{sec:mother}) which gives a slightly different perspective on \Cref{thm:mainSimple}.

 We also provide algorithms that are faster than the previous state of the art. 
 \begin{restatable}{theorem}{thmbetterColoring}\label{thm:fasterCol}
 	For any constant $\eps>0$, there is a deterministic \CONGEST algorithm that computes a $O(\Delta^{1+\eps})$-coloring in $O(\Delta^{1/2-\eps/2})+\logstar n$ rounds on any graph with maximum degree $\Delta$.
 \end{restatable}
Initiating \Cref{thm:fasterCol} with  $\eps=\log_{\Delta} k$ yields the following corollary.
\begin{corollary}
For any $1\leq k\leq \Delta$, there is a deterministic \CONGEST algorithm that computes an $O(k\Delta)$-coloring in $O(\sqrt{\Delta/k})+\logstar n$ rounds on any graph with maximum degree $\Delta$.
\end{corollary}
\paragraph{Ruling Sets.}
 For an integer $r\geq 1$, a \emph{$(2,r)$}-ruling set of a graph $G=(V,E)$ is a subset $S\subseteq V$ of the vertices that is an independent set and satisfies that for any vertex $v\in V$ there is a vertex $s\in S$ in hop distance at most $r$ \cite{awerbuch89}. Ruling sets and their extensions (larger distance between nodes in $S$) have played an important role as subroutines in several algorithms, e.g., \cite{awerbuch89,panconesi1992improved,GHKM18,EM19}. We provide a faster algorithm for  $(2,r)$-ruling sets. 
 \begin{restatable}{theorem}{thmRulingSets}\label{thm:rulingSet}
 	For any constant integer $r\geq 2$, there is a deterministic \CONGEST algorithm that computes $(2,r)$-ruling in $O(\Delta^{\frac{2}{r+2}})+\logstar n$ rounds on any graph with maximum degree $\Delta$. 
 \end{restatable} 
  The fastest previous algorithm used $O(\Delta^{2/r})+\logstar n$ rounds \cite{SEW13}. So, e.g., for $r=2$ the $\Delta$-dependency improves from $O(\Delta)$ to $O(\sqrt{\Delta})$ and for $r=3$ it improves from $O(\Delta^{2/3})$ to $O(\Delta^{2/5})$. 
 For $r=1$ the problem is equivalent to the \emph{maximal independent set problem} and has a $\Omega(\Delta)$ lower bound, if the dependency on $n$ is limited to $O(\logstar n)$ rounds \cite{FOCS19MIS}. For (possibly non constant) $r\geq 1$ there is a very recent lower bound of $\Omega(r\Delta^{1/r})$ rounds for the problem, even if an initial $O(\Delta)$ coloring is given \cite{BSKO21}. Thus, the bound of \Cref{thm:rulingSet} is tight for $r=2$. 

\paragraph{Lower bounds for color reduction.}
We give tight characterization for $1$-round color reduction algorithms, given an $m$-input coloring and no unique IDs.  
\begin{restatable}{theorem}{oneRoundTight}
	\label{thm:oneRoundTight}
	For any integer $\Delta\geq 1$ and $\Delta+1\leq m\leq \frac{\Delta^2}{4}+\frac{3\Delta}{2}+ \frac{9}{4}$ let $1\leq k\leq \min\{\Delta-1,\frac{\Delta}{2}+\frac{3}{2}\}$ be the largest integer such that $m\geq k(\Delta-k+3)$. Then, there is a $1$-round \CONGEST algorithm that on any $m$-input colored graph with maximum degree $\Delta$ computes an $(m-k)$-coloring. Additionally, there is no $1$-round \LOCAL algorithm that outputs a proper $(m-k-1)$-coloring on every $m$-input colored graph with maximum degree $\Delta$.
\end{restatable}
\Cref{thm:oneRoundTight} roughly states that reducing $k$ colors requires $k\Delta-\Theta(k^2)$ input colors. For concrete choices of $k$ the bound in \Cref{thm:oneRoundTight} says that to reduce $1$ color one needs at least $\Delta+2$ input colors, to reduce $2$ colors one needs $2\Delta+2$ input colors, to reduce $3$ colors one needs $3\Delta$ input colors, and to reduce $4$ colors one needs $4\Delta-4$ input colors, and so on \ldots~. 

The fastest randomized algorithms compute $O(\Delta)$ colorings in $O(\logstar n)$ rounds for $\Delta\geq \polylog n$ \cite{SW10,CLP20} and they can be adapted to also compute $(1+\eps)\Delta$-colorings. However, it seems that the hardest part of $(\Delta+1)$-coloring is to reduce a $(1+\eps)\Delta$ coloring to a $(\Delta+1)$-coloring. We show, that  an algorithm with runtime $T$ that reduces an input coloring with $(1+\eps)\Delta$ colors to a $(\Delta+1)$ coloring can be used with $O(\log_{1+\eps}\Delta)$ overhead to reduce a $O(\Delta^2)$-coloring to a $(\Delta+1)$-coloring. If $T=\Delta^{\Omega(1)}$, there is only a constant factor overhead (for details see \Cref{sec:conclusion}).

\subsection{Related Work}
The state of the art for $(\Delta+1)$-coloring when the runtime is expressed as $f(\Delta)+\logstar n$ is $O(\sqrt{\Delta\log \Delta})+\logstar n$ rounds and given by \cite{MT20}. Just, as the slightly slower algorithms \cite{FHK,Barenboim16,BEG18} the result of \cite{MT20} works for the more general $(deg+1)$-list coloring problem in which the size of the list of each node exceeds its degree. The result in \cite{FHK}  applies for the even more general \emph{local conflict coloring problem} in which one can specify for each edge of the graph which colors are not allowed to be adjacent. For an extensive overview on algorithms whose runtime is $f(\Delta)+\logstar n$  as well as an overview on the influence of arbdefective colorings during the last decade we refer to the related work section in \cite{MT20}. Further, almost all published papers until 2013 are discussed in the excellent monograph Barenboim and Elkin \cite{barenboimelkin_book}, and another very detailed overview on more recent results on coloring  is contained in \cite{Kuhn20}. Detailed overviews on randomized algorithms are contained in \cite{CLP20,HKMT21}.
Due to the sheer amount of published papers on distributed coloring we focus on selected results that have not been discussed in detail in \cite{barenboimelkin_book,Kuhn20,MT20,CLP20,HKMT21},  are most related to the current work, or indicate in which direction future research should continue, or should probably not continue. 

The objective in the edge coloring problem is to assign a color to each edge of a graph such that adjacent edges obtain different colors. Even though Vizing's Theorem \cite{vizing1964} states that any graph with maximum degree $\Delta$ can be colored with $\Delta+1$ colors and several randomized and deterministic  algorithms get close to this bound, e.g., \cite{GKMU17,HarrisEdge19,SuVu19,Bernshteyn22} in \LOCAL   and \cite{HMN22} in \CONGEST,  the classic objective is to use $2\Delta-1$ colors. The reason being that $(2\Delta-1)$-edge coloring is a $(\text{Maxdegree}+1)$-vertex coloring of the line graph.  The problem has a $f(\Delta)+\logstar n$ \LOCAL algorithm with $f(\Delta)=\poly\log \Delta$ \cite{BKO22a}. Obtaining such a runtime for computing a $(\Delta+1)$-vertex coloring would be a major breakthrough. In the \CONGEST model, the state of the art for $(2\Delta-1)$-edge coloring is an algorithm using $O(\Delta+\logstar n)$ round deterministic algorithm \cite{BEG18}.  An edge coloring with $O(\Delta)$ colors can be computed in the \CONGEST model with $f(\Delta)=\poly\log \Delta$ \cite{BKO22a}.

Besides optimizing the dependency on $\Delta$ after spending only $\logstar n$ rounds on Linial's coloring algorithm, another big branch of research has tried to settle the complexity of the problem as a function of $n$. For almost thirty years the best deterministic algorithm in this regime was a $2^{O(\sqrt{\log n})}\gg \poly\log n$ round algorithm \cite{awerbuch89,panconesi1992improved}, that has been improved to $O(\log^5 n)$ rounds \cite{RG19,GGR20} in the \LOCAL model  and in the \CONGEST model \cite{BKM19, GGR20} (slightly slower). The crucial building block of all of these results is a decomposition of the graph into $O(\log n)$ classes $\mathcal{C}_1,\ldots,\mathcal{C}_{O(\log n)}$ of small diameter clusters. To solve the $(\Delta+1)$-coloring problem one iterates through the $O(\log n)$ classes and solves each cluster $C\in \mathcal{C}_i$ in parallel in time that is (at least) linear in the cluster diameter. Even existentially, such decompositions require that the cluster diameter is at least $\Omega(\log n)$ and as a result these methods can probably not yield runtimes that are $o(\log^2 n)$. Thus, the fastest algorithm \cite{GK20} that needs $O(\log n\log^2\Delta)$ rounds uses a different approach: Similar to \cite{BKM19} it derandomizes a simple randomized $1$-round algorithm. Output colors are represented as bit strings of length $O(\log \Delta)$ and in one round of the algorithm each node  flips a (suitably weighted) coin to determine the next bit in the string. In expectation, after all $O(\log \Delta)$ bits are fixed a constant fraction of the vertices can be colored. Bamberger, Kuhn and Maus derandomize this algorithm for each cluster of a given network decomposition \cite{BKM19}. In contrast, instead of computing a network decomposition and derandomizing within a cluster, Ghaffari and Kuhn  derandomize the algorithm \emph{globally} with the help of a special kind of a defective coloring \cite{GK20}. Their derandomization step takes $O(\log \Delta)$ rounds for each of the $O(\log \Delta)$ bits and the $O(\log n)$ factor follows as only a constant fraction of the vertices get colored in each phase, yielding a total runtime of $O(\log^2\Delta\log n)$ rounds. Similar methods, also yielding $\log n \cdot \poly\log \Delta$ runtimes, have been successful for edge-coloring \cite{FOCS18-derand,HarrisEdge19} and computing maximal matchings \cite{F20}. 
If one allows $O(\Delta^{1+\eps})$ colors for a constant $\eps>0$ a $O(\log \Delta\log n)$ round algorithm has been known for more than a decade \cite{barenboimE10}.

As shown in \cite{brandt2016LLL,chang16exponential} a logarithmic dependency on $n$ ($\log \log n$-dependency for randomized algorithms) is unavoidable if one colors with fewer than $\Delta+1$ colors, that is, $\Delta$-coloring requires at least $\Omega(\log n)$ rounds. Similar bounds hold for the edge coloring problem for coloring with fewer than $(2\Delta-1)$-colors \cite{BHKOS19} and for coloring trees and bounded arboricity graphs with significantly fewer than $\Delta$ colors \cite{linial92,BE10sublog}. 

Little is known on lower bounds for $C$-coloring when $C\geq \Delta+1$ (in contrast to other symmetry breaking problems, e.g., maximal matching, MIS or ruling sets \cite{kuhn16_jacm,FOCS19MIS,BBO20}). 
Linial's $\Omega(\logstar n)$ deterministic lower bound has recently been re-proven in a  topological framework \cite{fraigniaud2020topology}. 
A $\Omega(\Delta^{1/3})$ lower bound for $O(\Delta)$-coloring holds in a weak variant of the LOCAL model  \cite{disc16_coloring}. 
Several papers 
analyzed special cases of coloring algorithms which can only spend a single communication round \cite{SzegedyV93,KuhnW06,disc16_coloring}. Just, as the lower bounds in this paper, none of these results gives anything non-trivial for two rounds. 
Also, the \emph{speedup} technique (e.g.,  \cite{Brandt19speedup,brandt2016LLL,FOCS19MIS,br2020truly,BBO20,balliu2019classification}), which proved very successful, e.g., for MIS and ruling set lower bounds, is not yet helpful for graph coloring. To make full use of the technique, one uses a computer program \cite{O20} to automatically transfer a problem $P_0$, e.g., the $(\Delta+1)$-coloring problem, into a problem $P_1$ that requires exactly one communication round less in the \LOCAL model. Then, one iterates the process to obtain problems $P_0,P_1,\ldots,P_t$, and if $P_t$ cannot be solved with a $0$-round algorithm, problem $P_0$ has a lower bound of $t$ rounds. Usually, the program is applied for small values of $\Delta$, and in a second step, the gained insights are transferred into a formal proof for general $\Delta$. For graph coloring the description of the problems grows so quickly with $t$ that even for small values of $\Delta$ one cannot even compute $P_1$ with current computers. 

There has also been a lot of progress in randomized coloring algorithm, e.g,   \cite{BEPSv3, HSinSu18, CLP20, HKMT21} where the state of the art for $(\Delta+1)$-vertex coloring is a $O(\log^3\log n)$ algorithm in the \LOCAL model \cite{CLP20,GK20} and $O(\log^6\log n)$ in the \CONGEST model \cite{HKMT21}.
Remarkably, there is a randomized $O(\logstar \Delta)$ round algorithm to compute a coloring with $\Delta+\log^{\gamma} n$ colors for a large enough constant $\gamma>0$  \cite{CLP20}. Prior to this Schneider and Wattenhofer \cite{SW10} showed that one can compute a $O(\Delta + \log^{1.1} n)$-coloring in $O(\log^* \Delta)$ rounds of \LOCAL.  Very recently, Halld\'orsson and Nolin showed that these results can be extended to the \CONGEST model \cite{HN21}. All of these latter randomized algorithm make use of the concept of trying several colors in one round, similar to our algorithm for $k>1$. 
In 2021, Halld\'orsson, Kuhn, Nolin, and Tonoyan \cite{HKNT21} showed that $(deg+1)$-list coloring can be solved in $\poly\log\log n$ rounds in the randomized \LOCAL model.  In 2022,  Halld\'orsson,  Nolin, and Tonoyan obtained similar result in the \CONGEST model \cite{HNT22}. These randomized algorithms run in $O(\logstar n)$ time if (additionally) it is guaranteed that each list is of size at least $\log^{2+\Omega(1)}n$), or $\Omega(\log^7 n)$, respectively, otherwise they run in $O(\log^3 n)$ \LOCAL rounds.
While Naor extended Linial's $\Omega(\logstar n)$ lower bound to randomized algorithms \cite{Naor91} (on rings with $\Delta=2$) it is not known whether the bounds for $O(\Delta)$-coloring for large $\Delta$ are tight. When $\Delta\geq \poly\log n$ holds,  our current knowledge does not rule out $O(1)$-round algorithms for $O(\Delta)$-coloring. This question is even more of interest as in this setting a $\poly \Delta$-coloring can be computed in one round from unique IDs from a space of size $\poly n$.

Additionally, we want to point out that, independently from this work, Barenboim, Elkin, and Goldenberg have extended their clever algorithm \cite{BEG18} to compute $O(k\Delta)$-colorings in $O(\Delta/k+\logstar n)$ rounds. These results appear in the journal version \cite{BEG-journal22} officially scheduled for publication next year.

\subsection{Roadmap}
In \Cref{sec:mother} we present the $O(\Delta/k)$-round $O(k\Delta)$-coloring algorithm, its implications and modifications to compute defective and out degree colorings. In \Cref{sec:applications} we explain how this simplifies the state of the art for $(\Delta+1)$-vertex coloring, our $O(k\Delta)$-coloring in $O(\sqrt{\Delta/k})$-rounds, and our results on computing ruling sets.
In \Cref{sec:oneRound} we analyze $1$-round color reduction algorithms. 
In \Cref{sec:conclusion} we conclude and explain why reducing a $(1+\eps)\Delta$-input coloring to a $(\Delta+1)$-coloring might be the hardest part of the $(\Delta+1)$-coloring problem. 

\section{Main Algorithm: Coloring Made Easy}
\label{sec:mother}
The objective of this section is to prove \Cref{thm:mainSimple} where the emphasis is on the fact that the algorithm is extremely simple if one ignores the precise choice of parameters. 
Before we prove \Cref{thm:mainSimple}, we prove \Cref{cor:allInOne} with useful settings of the respective parameters in \Cref{thm:mainSimple}; while \Cref{thm:mainSimple} is the technical result, the corollary is supposed  to be the framework to the outer world. To formally state the corollary, we begin with two definitions.

\begin{proof}[Proof of \Cref{cor:allInOne}]
	In each of the results we apply \Cref{thm:mainSimple} with different parameters. 
\textbf{Proof of (1).}  Choose $d=0$ which implies $X=16\Delta$. With $k=X$ we obtain a proper $C=X\cdot k=256\Delta^2$-coloring in one ($R=X/k=1$) round. 
\textbf{Proof of (2).} Choose $d=0$, which implies $X=16\Delta$. Then we obtain a $16\Delta\cdot k$ coloring in $R=16\Delta/k$ rounds. 
\textbf{Proof of (3).} Choose all parameters as in 2., but set $k=\lceil\Delta/16\rceil$, which implies $\Delta^2$ colors in $R=16\Delta/k=O(1)$ rounds. 

\smallskip 

In (4)--(6). the condition $d=\beta=\Delta^{\eps}$ implies \begin{align*}
X=O(\Delta/(\beta+1)\cdot \log_{\Delta/(\beta+1)}\Delta^4)=O(\Delta/\beta).
\end{align*}

\textbf{Proof of (4).} Let $k=1$ which implies the claimed number of colors ($X\cdot k=O(\Delta/(\beta +1)$)) and the claimed round complexity $(R=O(\Delta/\beta)$. The coloring is a $\beta$-outdegree coloring due to part (1) of \Cref{thm:mainSimple}. 
\textbf{Proof of (5).} With $k=X$ the runtime is $R=X/k=1$ rounds and we obtain $C=X\cdot k=O((\Delta/d)^2))$ colors. \Cref{thm:mainSimple} says that the coloring has defect at most $d$ as there is only one subgraph ($R=1$). 
\textbf{Proof of (6).} Choose $k=1$. 
	Let $P_1,\ldots, P_R$ be the partition of part $(2)$. If vertices consider their color and the index of their part of the partition as a color tuple, i.e., if a vertex $v\in P_j$ with color $\phi(v)$ colors itself with color $(\phi(v),j)$ we obtain a $d$-defective coloring with $O(\big(\frac{\Delta}{d}\big)^2)$ colors in $O(\Delta/d)$ rounds. 
\end{proof}
\Cref{cor:allInOne} shows that one algorithm is sufficient for many of the essential steps of several previous important papers,  and it further allows to smoothly scale between these results. Additionally, the algorithm for computing $\beta$-outdegree colorings is simpler and more direct than previous algorithms. The algorithm in \cite{BEG18} first needs to computes a certain defective coloring and only then can proceed to compute a low outdegree coloring. The slower algorithm in \cite{Barenboim16} uses a more involved recursive approach. The algorithm(s) in \cite{BE10sublog,BE11} are more involved and require $\Omega(\log n)$ rounds. 

We continue with explaining the algebraic basics to construct the sequences for the algorithm for \Cref{thm:mainSimple} (Algorithm \ref{alg:mother}). 
Given a prime $q$ let $\FF_q$ denote the field of size $q$ over the elements $[q]=\{0,\ldots,q-1\}$ and let 
\begin{align*}
	P_q^f=\{p:\FF_q\rightarrow \FF_q \mid \text{ $p$ is polynomial of degree $\leq f$}\}
\end{align*} 
be the set of all polynomials over $\FF_q$ of degree at most $f$. Recall that $Z=\Delta/(d+1)$. 
To run the algorithm on a graph with maximum degree $\Delta$, an input $m$-coloring and a \emph{defect parameter} $d$ fix  $f=\lceil \log_{Z} m\rceil $ and a prime $q$ with 
\begin{align}
	\label{eqn:choiceOfQ}
	2f\cdot Z<q<4f\cdot Z,
\end{align}
which exists due to Bertrand's postulate.  Then we can locally and without communication assign each input color $i\in [m]$ a distinct polynomial $p_i\in P_q^f$  as $m\leq |P_q^f|=q^{f+1}$ and since all vertices know $m$ and $f$. For example, we can represent each element $p(x)=\sum_{i=0}^f a_i x^i $ of $P_q^f$ as a tuple $(a_0,\ldots,a_f)$, order the tuples lexicographically and assign the polynomial corresponding to the $i$-th tuple with input color $i$.  Given these polynomials, uncolored nodes try the colors in their sequences in batches of size $k$ ($k$ is an integer parameter that can be freely chosen). To \emph{try a color} it is sent to its neighbors.  A node gets permanently colored with a color $c$ or \emph{adopts} color $c$ if it causes conflicts with at most $d$ neighbors. To this end call a color $c$ \emph{$d$-proper} in some iteration if at most $d$ neighbors try color $c$ in the same iteration or are already permanently colored with color $c$.  The details are given by Algorithm~\ref{alg:mother} and the following paragraph.

\RestyleAlgo{boxruled}
\begin{algorithm} \caption{For vertex with color $i$. Parameters $d, k, m, \Delta$. (\LOCAL model)} 
	\label{alg:mother}
	\textbf{Locally compute:}\\
	\SetInd{0.1em}{2em}
	\quad polynomial $p_i:\mathbb{F}_q\rightarrow \mathbb{F}_q$ with $q$ chosen by (\ref{eqn:choiceOfQ}) \\
	\quad  sequence $s_i$: $\big(x\mod k, p_i(x) \mod q \big)$,~~$x=0,\ldots,q-1$\\
	\quad  Split $s_i$ into $\big\lceil \nicefrac{q}{k}\big\rceil$ consecutive batches $B_1,B_2, \ldots$, each of size $k$ (except for the last one)\\
	\textbf{For } $j=1,\ldots, \big\lceil \nicefrac{q}{k}\big\rceil$ \\
	\quad	\emph{Try} the colors in batch $B_j$ (in a single round)\\
	\quad	\lIf{ $\exists$ \big($d$-proper $c\in B_j$\big)}
	{adopt $c$, join $P_j$, and \Return } 
\end{algorithm}

Sending the colors of one batch takes one round of communication in the \LOCAL model. We reason at the end of the proof of \Cref{thm:mainSimple} that processing one batch also can be done in a single round of \CONGEST.  
If $k$ does not divide $q$, a node that is uncolored before the last iteration tries less than $k$ tuples in the last iteration, i.e., $|B_{\lceil q/k\rceil}|<k$. In fact, it will try $q-k\lfloor q/k\rfloor$ tuples. When a vertex picks a tuple as its permanent color $c$, it orients all edges towards neighbors that have previously chosen $c$ as a permanent color. If two neighbors both pick the same permanent color $c$ in the same iteration, the edge between them is oriented arbitrarily (e.g., using the input coloring for symmetry breaking from smaller input color to larger input color). A node joins the subgraph $P_j$ where $j$ is the index of the iteration in which it decides for a permanent color.  

Note that vertices with the same input color compute the same sequence and that all steps that are related to orientations and subgraphs are obsolete when $d=0$. 
We will show that the algorithm is well defined, i.e., that every vertex is colored before it reaches the end of its sequence. To prove the result, we need the following well known algebraic result on the number of intersections of two degree bounded polynomials over finite fields. 
\begin{lemma}
	\label{lem:polyIntersect}
	Let $q$ be a prime, $f\in \NN_0$ and let $p_1,p_2\in P_q^f$ distinct polynomials of degree $f_1, f_2$, respectively. Then there are at most $\max\{f_1,f_2\}$ points in which $p_1$ and $p_2$ intersect, i.e., $|\{x\in \FF_q \mid p_1(x)=p_2(x)\}|\leq \max\{f_1,f_2\}$~. 
\end{lemma}

\begin{proof}[Proof of \Cref{thm:mainSimple}] Recall, that the prime number $q$ is the size  of the field $\FF_q$ from which the coefficients of the polynomials are taken. Due to the choice of $q$ and $f$ (see \Cref{eqn:choiceOfQ}) the set $P_q^f$ contains at least one distinct polynomial for each input color $i\in[m]$. Recall, the choice of parameters: $Z=\frac{\Delta}{d+1}$, $X=4\cdot Z\cdot\lceil\log_{Z}m\rceil$, and $R=\Delta/k$. Let $C=X\cdot k$ and note that $X\geq q$. 	We first prove all statements under the assumption that all vertices are colored after the $R$ iterations of the  loop; afterwards we show that this is indeed the case.
	
	\textbf{Bounding $\#$ colors: }Each color is of the form $(x \mod k, p(x)\mod q)$ for some polynomial $p$ that is evaluated over $\FF_q$. Thus the number of colors $C$ can be upper bounded by $k\cdot q\leq k\cdot X$.	
	
	\textbf{Proof of $(1)$:} When a vertex $v$ is colored with a color $\phi(v)$ in iteration $j$ there are at most $d$ other vertices that try to get the same color in this round or are already colored with this color. Since only edges to these vertices are oriented outwards from $v$ the outdegree of $v$ is bounded by $d$. Further, each edge between vertices with the same permanent color is oriented:  Either they picked the color in the same iteration and the edge is oriented from the node with the smaller input color to the vertex with larger input color, or the edge is oriented outwards from the vertex that got permanently colored in a later iteration. 
	
	\textbf{Proof of $(2)$:} A vertex $v$ joins $P_j$ if it is colored in iteration $j$ of the loop. As a vertex $v$ only gets colored with a color $\phi$ in iteration $j$ if there are at most $d$ neighbors of $v$ that want to get colored with color $\phi$ in iteration $j$ the maximum degree of the graph induced by all vertices in $P_j$ with color $\phi$ is at most $d$. 
	
	\textbf{All vertices are colored after the $R$ iterations of the  loop:}
	A node is not colored in one iteration only if for all of the $k$ tuples, i.e., colors of the form $(x,p_i(x))\in [k]\times [q]$, that it tries  in that iteration there are strictly more than $d$ neighbors that try the same tuple in the current iteration or are already colored with the tuple. Before we conclude with the proof that all vertices are colored at the end, we bound the number of conflicts that a node experiences during the execution of the algorithm. We consider two types of conflict. 
	
	\textit{Bounding the number of conflicts with an active node by $f$:} Let us bound the number of times in which two neighbors $u$, $v$ with polynomials $p_u$ and $p_v$, respectively,  try the same tuple in some iteration $j\in [\lceil q/k\rceil]$ (conditioned on both nodes not being permanently colored yet). In iteration $j$, the nodes simultaneously try all of the following tuples (where we omit the $j\cdot k\mod k=0$ term in the first coordinate). 
	\begin{align}
		\text{$u$ tries: } & \big(l , p_u(j\cdot k+l)\big) \text{ with $l\in [k]$} \text{ and}\\
		\text{$v$ tries: }& \big( l, p_v(j\cdot k+l)\big)\text{ with $l\in [k]$.} 
	\end{align} 
Two tuples tried by $u$ and $v$ in iteration $j$ can only cause a conflict, i.e, be the same, if they are the same in both coordinates. As all $k$ tuples that are simultaneously tried by a node differ in the first coordinate,  any conflict in iteration $j$ between $u$ and $v$ implies that $p_u(j\cdot k +l)=p_v(j\cdot k+l)$ holds for some $l\in [k]$.  Since $p_u$ and $p_v$ are polynomials of degree at most $f$, \Cref{lem:polyIntersect} implies that there are at most  at most $f$ combinations of $j$ and $l$ for which this holds.
	
	\textit{Bounding the number of conflicts with an inactive node by $f$:} Consider a neighbor $u$ that chose some permanent color $\big(x_u,y_u\big)\in[k]\times[q]$. For $v$ to try this tuple in iteration $j\in [q/k]$ we need \begin{align*}
		\big((j\cdot k +l)\mod k,p_v(j\cdot k+l)\mod q\big)=\big(x_u,y_u\big)
	\end{align*}
 for some $l\in [k]$. This can only be the case if $p_v(j\cdot k+l)$ equals the fixed number $y_u$, which is the case for at most $f$ different choices of $j$ and $l$ due to \Cref{lem:polyIntersect} ($y_u$ is a polynomial of degree $0$).
	
	\smallskip
	
	Thus, for fixed $u$ and $v$, there are at most $f$ tuples causing a \emph{conflict} while $u$ and $v$ are active, and at most $f$ tuples causing a \emph{conflict} after (at least) one of the nodes has chosen a permanent color. 
	A node $v$ cannot get permanently colored with a tuple $\big(l,p_v(j\cdot k+l)\big)$ if there are  strictly more than $d$ conflicts for the tuple, i.e., strictly more than $d$ neighbors try the same tuple in the same iteration or have already permanently adapted the color. In this case we call the tuple \emph{blocked}. 	
	As each of the at most $\Delta$ neighbors contributes at most $2f$ such conflicts there can be at most $z=2f\cdot Z$ blocked tuples. As the length $q$ of the sequence (of tried tuples) is strictly larger than $z$, there is at least one tuple that is not blocked and each node is colored at the end of the algorithm.

	\textbf{\CONGEST implementation:} During the execution of the algorithm all nodes have knowledge of $m$, $q$, $f$, $d$ and $k$ and all nodes can construct the set of polynomials $P_q^f$ locally according to the same lexicographic order. Thus, for an uncolored node to send $k$ trials in iteration $j$, it is sufficient to send its input color (together with $k$ and $j$ which are global knowledge). A node that gets colored can inform its neighbors about the choice in one round. Hence, a node never needs to send more than a single color per round and the message size is upper bounded by $O(\log \Delta)=O(\log n)$ bits and can be executed in the \CONGEST model.
\end{proof}
We remark, that it becomes significantly easier to read the algorithm if $m=\poly\Delta$ and if $\Delta/d=\Delta^{\Omega(1)}$, as this implies $f=O(1)$, which simplifies many of the parameters. However, we chose to present the result in a more general form.

\begin{remark}With a tighter analysis for special cases one can reduce the constants in \Cref{cor:allInOne}, e.g., in the case of $k=X$, the size of the field $\FF_q$ can be chosen smaller.  Due to such a tighter analysis and by assuming $m=\Delta^3$ the leading constant in the $O(\Delta^2)$-coloring by Linial is some  $1<\alpha<10$ \cite{linial92}.  In contrast, the lower bound for the one-round algorithms from \Cref{sec:oneRound} only provides impossibilities below $\Delta^2/2+\Theta(\Delta)$ colors, that is, for a constant $\alpha<1$. Thus, there is a large regime for $\alpha$ where we neither have one-round upper bounds nor lower bounds. As even optimized constants in \Cref{thm:mainSimple} and \Cref{cor:allInOne} leave a gap for the regime of $\alpha$ where lower bounds are known, we focus on having simple proofs that cover all cases of the theorem, instead of optimizing these constants.
\end{remark}
The condition  $d=\beta=\Delta^{\eps}$ in \Cref{cor:allInOne}: One would wish to use a variant of \Cref{cor:allInOne} to compute a $\Delta/2$-defective $O(1)$-coloring in one round, given a $O(\Delta^2)$-coloring. Note, that this setting would require the finite field over which we operate---the field size essentially determines the number of colors---to contain only $q=O(1)$ elements, and to obtain a distinct polynomial for each input color we would have to choose polynomials of degree $f=O(\log \Delta)\gg q$. This immediately violates \Cref{eqn:choiceOfQ}. Also, in that case, the proof of \Cref{thm:mainSimple} breaks as we might have $O(f\Delta/d)=O(f)=\omega(1)$ blocked tuples while only having $q=O(1)$ tuples in the sequence. While slightly weaker requirements on $d$ are possible without breaking the proof,  our requirement ensures that we do not run into these issues. See \cite{Kuhn2009WeakColoring,BarenboimEK14} for the parameter-heavy details on how to iterate the result of \Cref{thm:mainSimple} for $O(\logstar \Delta)$ iterations to obtain defective a $d$-defective $O((\Delta/d)^2)$-coloring with no condition on $d$ (essentially \Cref{cor:allInOne} (5) can also take a defective coloring as input coloring, and then defects add up). 

We point out that the sequences required for \Cref{thm:mainSimple} need not be constructed via polynomials. The proof only requires that the elements of the sequence are from a small enough domain, sequences are long enough, there is one sequence for each input color, and any two sequences intersect in few positions. In \cite[arxiv version]{MT20} such sequences are constructed greedily. Here, we chose to use a construction based on polynomials as the dependency on the $m$-input coloring is better, in particular, when $m=\poly \Delta$, it implies that $f=O(1)$, instead of $f=O(\log \Delta)$ for the greedy-based construction. 
\section{Applications}
\label{sec:applications}
\subsection{Simplifying $(\Delta+1)$-Coloring Algorithms}
\label{ssec:simplified}
\Cref{cor:allInOne} provides a simpler algorithm to compute so called $\beta$-outdegree colorings with $O(\Delta/\beta)$ colors in $O(\Delta/\beta)$ rounds matching the state of the art of \cite{BEG18}. Recall, that an $\beta$-outdegree $c$-coloring is an improper coloring with $c$ colors in which each monochromatic edge is equipped with an orientation such that the outdegree of each vertex is at most $\beta$. These colorings are used in all known $(\Delta+1)$-coloring algorithms whose round complexity is sublinear in $\Delta$ \cite{Barenboim16, FHK, BEG18, MT20}, and hence our simplifications carry over to these algorithms. The core message of our work is that all results in \Cref{cor:allInOne} can be obtained through modifications of Linial's algorithm. Next, we sketch the algorithm in \cite{MT20} to explain that the same is true for their result. The algorithms of \cite{Barenboim16,FHK,BEG18,MT20} use the following high level scheme: 
First compute a $\beta$-outdegree $z$-coloring for a suitable choice of $\beta$ such that $z=O(\Delta/\beta)=o(\Delta)$. Its color classes yield a partition $V_1,\ldots, V_{z}$ of the vertex set. In a second step, the partition is used as a schedule. In order to compute the nodes' final output color we iterate through the schedule. For the purpose of this exposition we may assume that all nodes of $V_i$ are colored with their final output color after processing them (formally, this is only true for certain nodes of $V_i$ and an additional recursion is required to color all nodes of the graph). When processing $V_i$, we ensure that a node does not get colored with a color of any of its already colored neighbors in $V_1\cup \ldots \cup V_{i-1}$, that is, when a node is processed it does not have all of the $\Delta+1$ output colors available, but instead its \emph{list of available colors} does not include the colors of its already colored neighbors. Thus, the resulting problem that we need to solve on $G[V_i]$ is a so called \emph{list coloring} problem.  The additional lever, when coloring $V_i$,  is that $G[V_i]$  is equipped with an orientation with a small outdegree. The core result of the paper with the title \emph{'Linial for Lists'} \cite{MT20} is a generalization of Linial's $1$-round color reduction algorithm to the list coloring problem that works in two rounds in graphs with small outdegrees. Hence, the crucial coloring step as well as the algorithm to compute the necessary schedule, i.e., the $\beta$-outdegree coloring, are generalizations of Linial's algorithm.

\subsection{Improved $O(\Delta^{1+\eps})$-Coloring Algorithms}
If one aims for $O(\Delta)$ colors (instead of $\Delta+1)$ colors) the  scheme  that we explained in \Cref{ssec:simplified} works for a suitable choice of $\beta=\Theta(\sqrt{\Delta})$, yielding a runtime of $O(\Delta/\beta)=O(\sqrt{\Delta})$ rounds.  
\begin{theorem}[\mbox{\cite{Barenboim16,BEG18}}]
	\label{thm:BarenboimCongest}
	There is a deterministic \CONGEST algorithm that computes a $O(\Delta)$-coloring in $O(\sqrt{\Delta}+\logstar n)$ rounds on any graph with maximum degree $\Delta$.  
\end{theorem}
We use this algorithm and our defective coloring algorithm from \Cref{cor:allInOne} to improve the trade-off between the number of colors and the runtime from $O(k\Delta)$ vs $O(\Delta/k)$ (\Cref{cor:allInOne}) to $O(k\Delta)$ vs $O(\sqrt{\Delta/k})$. 

\thmbetterColoring*
\begin{proof}
	First, compute an $O(\Delta^2)$-coloring in $\logstar n+O(1)$ rounds using Linial's algorithm \cite{linial92}.
Set $d=\Delta^{1-\eps}$, then use \Cref{cor:allInOne} (part 6) to compute a $d$-defective coloring $\psi$ with $O((\Delta/d)^2)$ colors in $O(\Delta/d)=O(\Delta^{\eps})$ rounds. Then, on each color class in parallel compute a $O(d)$-coloring in $O(\sqrt{d})=\Delta^{1/2-\eps/2}$ rounds (the $\logstar n$ of \Cref{thm:BarenboimCongest} vanishes as we already have a $O(\Delta^2)=\poly d$-coloring) via \Cref{thm:BarenboimCongest} using a distinct color space for each color class of $\psi$, that is, each node $v$ gets a color $\phi(v)$ from this second step and the final output color of node $v$ is set to be the tuple $(\psi(v),\phi(v))$. In total we use $O((\Delta/d)^2\cdot d)=O(\Delta^2/d)=O(\Delta^{1+\eps})$ colors. 
\end{proof}
While the above way is a simple way to prove \Cref{thm:fasterCol} when using \Cref{thm:BarenboimCongest} as a blackbox, an alternative algorithm can be obtained by using the $\beta$-outdegree coloring result from \Cref{cor:allInOne} (with $\beta=\Delta^{1/2+\eps/2})$ and carefully choosing the remaining parameters  in the framework of \cite{Barenboim16}. The corresponding parameters do not appear in the scheme in \Cref{ssec:simplified}. 

\subsection{Ruling Sets}
A \emph{$(2,r)$}-ruling set of a graph $G=(V,E)$ is a subset $S\subseteq V$ of the vertices that is an independent set and satisfies that for any vertex $v$ in $V$ there is a vertex $s\in S$ in hop distance at most $r$ \cite{awerbuch89}. The following result uses colorings to compute a ruling sets. 
\begin{lemma}[\mbox{\cite[arxiv version]{KMW18}}]
	\label{lem:rulingSet}
For any $B\geq 2$ there exists a deterministic distributed \CONGEST 
algorithm that, given a $C$-coloring, computes a
$(2,\lceil\log_B C\rceil)$-ruling set in $O(B\log_B C)$ rounds. 
\end{lemma}
We use \Cref{lem:rulingSet} and \Cref{thm:fasterCol} to compute $(2,r)$-ruling sets by adjusting the number of colors such that the runtime of computing the coloring and using it via \Cref{lem:rulingSet} are balanced. 
\thmRulingSets*
\begin{proof}
	Set $\eps=\frac{r-2}{r+2}$, and use \Cref{thm:fasterCol} to compute an $O(\Delta^{1+\eps})$-coloring in $O(\Delta^{1/2-\eps/2})+\logstar n$ rounds. Let $C=O(\Delta^{1+\eps})=O(\Delta^{\frac{2r}{r+2}})$ be the number of colors of this coloring. 
 Now, set $B$ such that $\lceil\log_B C\rceil=r$ and apply \Cref{lem:rulingSet} to compute a $(2,r)$-ruling set in $O(B\log_BC)=O(B\cdot r)=O(C^{1/r})$ rounds. 
Ignoring the $\logstar n$ term, the total runtime is upper bounded by
\begin{align*}
O(\Delta^{\frac{1}{2}-\frac{r-2}{2r+4}}+C^{1/r}) & = O(\Delta^{\frac{1}{2}-\frac{r-2}{2r+4}}+ \Delta^{\frac{2}{r+2}}) 
  =O(\Delta^{\frac{2}{r+2}})~. & & \qedhere
\end{align*}
\end{proof} 
Interestingly, for $r=2$ the state of the art runtimes for $(2,r)$-ruling set is the same as the complexity for computing an $O(\Delta)$ coloring. Note that the runtime bound of \Cref{thm:rulingSet} cannot be achieved for $r=1$.  In that case, the ruling sets are better known under the name \emph{maximal independent sets} for which a $\Omega(\Delta)$-round lower bound is known if the runtime's $n$-dependency is limited to $O(\logstar n)$ \cite{FOCS19MIS}.

An \emph{$(\alpha,r)$-ruling set} is a subset $S\subseteq V$ that is  an independent set in the power graph $G^{\alpha-1}$ that satisfies that each vertex $v\in V$ has a vertex in $S$ in distance at most $r$. In the \LOCAL model the results of \Cref{thm:rulingSet} can be extended to $(\alpha,r)$-ruling sets as one can simulate any algorithm on $G^{\alpha-1}$ in the original network graph. For details on these black-box extensions see, e.g.,  \cite{BEPSv3,KMW18}.

\section{One-Round Color Reduction}
\label{sec:oneRound}
The objective of this section is to show the following theorem.
\oneRoundTight*

In addition to the values for $k$ that are stated in \Cref{sec:intro}  ($k=1,2,3,4$) note that one requires $5\Delta-10$ input colors to reduce $5$ colors, and $6\Delta-18$ input colors to reduce $6$ colors. 
The proof is split into two lemmas. In \Cref{lem:colorReduction},  we provide a $1$-round color reduction algorithm, and in \Cref{lem:oneRoundLB} we show that the algorithm is tight up to each single color.
For a function $f:V\rightarrow [m]$ and a set $S\subseteq V$ we denote $f(S)=\{f(v)\mid v\in S\}$. For a node $v\in V$ of a given graph $G=(V,E)$ we denote the set of its neighbors by $N(v)$.

The following result reduces more colors than the $1$-round algorithms in \cite{KuhnW06,SzegedyV93}. 
\begin{lemma}[Color Reduction]\label{lem:colorReduction}
	For an integer $1\leq k\leq \frac{\Delta}{2}+\frac{3}{2}$ there is a $1$-round color reduction procedure from $m\geq k(\Delta-k+3)$ to $k(\Delta-k+2)$ (reduces $k$ colors). 
\end{lemma}
The intuitive idea of the algorithm is that only vertices in  the $k$ largest color classes change their color. Each of these \emph{recoloring colors} has its own small hardcoded output color regime from which it can pick a \emph{free color}. However, the size of the regime is smaller than $\Delta$ and it might be that all of its colors are blocked by neighbors that do not recolor themselves. But this implies that there are some recoloring colors that do not appear in the node's neighborhood and it can steal colors from those recoloring color's regimes to gain the desired freedom (one stolen color per regime). 

In the rest of the section, for an integer $x$, we use the notation $[x]$ to refer to the set $\{0,\ldots,x-1\}$.  
\begin{proof}[Proof of \Cref{lem:colorReduction}]
We may assume that $m=k(\Delta-k+3)$ is the number of input colors as for $m'\geq m$ input colors, one can leave $m'-m$ colors unchanged and apply the algorithm to the remaining $m$ colors. Let $\ell=k(\Delta-k+2)\geq \Delta+1$ be the number of (desired) output colors, $\phi:V\rightarrow [m]$ the input coloring and $\psi:V\rightarrow [\ell]$ the to be computed output coloring. 
	
	Fix  $k$ disjoint color regimes of output colors $R_0,\ldots R_{k-1}$, each of size $\Delta-k+2$ as follows: For $0\leq i<k$ and $j\in [\Delta-k+2]$, let $r_i(j)= i\cdot(\Delta-k+2)+j\in [\ell]$ and $R_i=\{r_i(j) \mid j\in [\Delta-k+2]\}\subseteq [\ell]$.

We refer to $R_i$ as the \emph{$i$-th regime}. The regimes are disjoint and each regime $R_i$ is of size $|R_i|=\Delta-k+2\geq k-1$ because $k\leq \Delta/2+3/2$ holds. Additionally, for each of the $k$ regimes, let $f_i: \{\ell, \ldots, \ell+i-1, \ell+i+1, m\}\rightarrow R_i$ be an abritrary injective function into the regime. Note that $f_i$ exists as its domain is of size $m-\ell-1=k-1\leq |R_i|$. Nodes execute the following algorithm that takes $1$ round as a node needs to learn its neighbors input colors. 
	
	\begin{algorithm}[H]\caption{Executed at each node $v$, output coloring $\psi$,  input coloring $\phi$}\label{alg:trycolor}
		\textbf{Send} input color $\phi(v)$ to neighbors; \textbf{Receive} set of neighbors' input colors $\phi(N(v))$;\\
		\textbf{Case $\phi(v)< \ell$:} $\psi(v):=\phi(v)$; \textbf{exit}; \\
		\textbf{Case $\max_{u\in N(v)}\{\phi(u)\}<\ell$:}  $\psi(v):=\min ([\Delta+1]\setminus \phi(N(v)))$; \textbf{exit}; \\
		\textbf{Case (else): }	
		\quad $F(v):= R_{\phi(v)}\cup\{f_j(\phi(v))\mid 0\leq j<k, \ell+j\notin \phi(N(v))\}$ \\
		\hspace{2.9cm} $\psi(v):=F(v)\setminus \phi(N(v))$
		
	\end{algorithm}
	Let $v$ be a node. If $v$ executes the first case, it does not change its color and neighbors ensure to not output the same color as $v$. If $v$ executes the second case, all of its neighbors execute the first case and do not change their colors; $v$ selects a color not conflicting with any of these.  Hence, for the rest of the proof assume that $v$ executes the third case. 
	By the next lemma we have $F(v)\cap F(w)=\emptyset$ for any neighbor $w$ of $v$, regardless of which case $w$ executes.

	\begin{claim}
		\label{claim:NoClashCase3}
		 For any two neighbors $v$ and $w$ we have $F(v)\cap F(w)=\emptyset$.
	\end{claim}
	\begin{proof}
		By definition and as $\phi(w)\neq \phi(v)$, the set $F(w)$ does not contain any color of $R_{\phi(v)}$. Similarly, the set $F(v)$ does not contain any color of $R_{\phi(w)}$. Hence, if $F(v)$ and $F(w)$ intersect, the intersection must lie in some regime $R_j$, where $j\neq \phi(v)$ and $j\neq \phi(w)$. However, $F(v)$ and $F(w)$ only contain a single color each in such an $R_j$. These colors are $f_j(\phi(v))$ and $f_j(\phi(w))$, respectively. We obtain $f_j(\phi(v))\neq f_j(\phi(w))$ because  $f_j$ is injective and $\phi(v)\neq \phi(w)$ holds. 
	\end{proof}
	 As $v$ executes the third case, we obtain $\phi(v)\geq\ell$, and no neighbor of $v$ can execute the second case. Hence, all neighbors of $v$ either stick to their color that is already $<\ell$, or also execute the third case. Let $d(v)$ be the number of neighbors of $v$ that execute the first case, i.e., do not recolor themselves. As $v$ does not use execute the second case,  we obtain $d(v)<\Delta$. 
	In order to show that $v$ can always select a color, i.e., that   that $F(v)\setminus \phi(N(v))\neq \emptyset$ holds, we show that  that $|F(v)|\geq d(v)+1$ holds.  We lower bound the size of $F(v)$ in two cases. 
	
	\textit{Case $\Delta-d(v)\leq k-1$: } 	Let $Y=\{\phi(w)\mid w\in N(v), \text{$w$ executes the third case}\}$ be the set of input colors of neighbors of $v$ that execute the third case. We obtain $|Y|\leq \min\{k-1,\Delta-d(v)\}$ (the $k-1$ term appears as the input color of node $v$ cannot appear as an input color in its neighborhood).
	 Let $X=[m]\setminus ([\ell]\cup Y\cup \{\phi(v)\})$ be the set of input colors $\geq \ell, \neq \phi(v)$ that do not appear in the neighborhood of $v$. We have $|X|= (k-1)-|Y|\geq  (k-1)-\min\{k-1,\Delta-d(v)\}=d(v)+k-\Delta-1$. For each color in $X$ there is a corresponding regime from which $F(v)$ contains one color. We obtain. 
	\begin{align*}
	|F(v)|\geq |R_j|+|X| & \geq \Delta-k +2 + d(v)+k-\Delta-1=d(v)+1.
	\end{align*}

	\textit{Case $\Delta-d(v)> k-1$: } The condition implies that $d(v)<\Delta-k+1$ holds and we obtain 
		\begin{align*}
		|F(v)|\geq |R_j|=\Delta-k+2> d(v)+1.
	\end{align*}

In both cases we obtain $|F(v)|\geq d(v))+1$. Thus, in the last line of the algorithm,  $v$ can pick a color in $F(v)$ not used by the $d(v)$ neighbors that do not recolor themselves, and due to \Cref{claim:NoClashCase3}  there is no conflict with any neighbor that executes the third case. To finish the proof, recall that no neighbor executes the second case.
\end{proof}
Next, we show that the result of \Cref{lem:colorReduction} is tight. The next lemma can be seen as a generalization of a result in \cite{disc16_coloring}), which proved a result of a similar flavor but only for $m\geq \frac{\Delta^2}{4}+\frac{\Delta}{2}+1$.
\begin{lemma}[Lower Bound for $1$-Round Algorithms]
	\label{lem:oneRoundLB}
	Let $k,\Delta$ and $m$ be arbitrary integers satisfying $1\leq k\leq \Delta-1$ and $m\leq k(\Delta-k+3)-1$. Then, 
	there is no $1$-round \LOCAL algorithm that computes a $q=m-k$ coloring on every $m$-input colored graph with maximum degree $\Delta$.
\end{lemma}
\begin{proof}
	We  may assume that $m=k(\Delta-k+3)-1$, as an impossibility result for this choice of $m$ implies the same result for any $m'\leq m$. 	Assume for contradiction that there is a $1$-round algorithm $\mathcal{A}$ that colors a graph with an input coloring with $m$ colors with $q=m-k$ output colors. 
	Call an input color $\phi$ \emph{sensitive} if for any output color $c\in [q]$, there is an input color $\phi'\in [m]$ such that if a node $v$ is input colored with $\phi$ and has a neighbor with input color $\phi'$ then node $v$ does not output $c$ (regardless of the colors of other neighbors). In the rest of the proof, our notation identifies nodes with their input color. 
	\begin{claim}
		\label{claim:sensitive}
		There are at least $k$ sensitive input colors. 
	\end{claim}
\begin{proof} We introduce the following definition where $c\in [m]$ is an output color. An input color $\phi$ is called \emph{$c$-robust} if for all input colors $\phi'\neq \phi$ there is a set $A$ of size at most $\Delta$ satisfying $\phi'\in A$ and $\mathcal{A}\big((\phi, A)\big)=c$. 
		Let $S=\bigcup_{c\in [q]}\{\phi\in [m]\mid \text{$\phi$ is $c$-robust}\}$ be the set of input colors that are $c$-robust for some $c$. Next, we upper bound the size of $S$ by $q$. In particular, we show that each of the sets in the union contains at most a single element. Assume for contradiction that $\phi$ and $\phi'$ are two distinct input colors that are $c$-robust. By definition, there are sets  $A_{\phi}$ and $A_{\phi'}$ of size at most $\Delta$, satisfying $\phi'\in A_{\phi}$ and $\phi\in A_{\phi'}$ such that $c=\mathcal{A}\big((\phi, A_{\phi})\big)=\mathcal{A}\big((\phi', A_{\phi'})\big)$, a contradiction to the correctness of $\mathcal{A}$ as $(\phi, A_{\phi})$ and $(\phi', A_{\phi'})$ can be neighborhoods of neighboring nodes in some graph. 

	Let $T=[m]\setminus S$ and observe that $|T|\geq m-|S|\geq k$. For each $\phi\in T$ and for all $c\in [q]$ $\phi$ is  not $c$-robust, that is, for all $c$ there exists an input color $\phi'$ such that for all sets $\phi'\in A$ of size at most $\Delta$, $\mathcal{A}(\phi,A)\neq c$. Hence, each color in $T$ is sensitive and the claim follows. 
\renewcommand{\qed}{\ensuremath{\hfill\blacksquare}}
\end{proof}
\renewcommand{\qed}{\hfill \ensuremath{\Box}}

	Next, we construct a $1$-hop neighborhood that cannot be colored with one of the $q$ output colors and thus leads to a contradiction. Let $T$ be a set of sensitive input colors of size $k$ which exists due to \Cref{claim:sensitive}.
	  Consider the partial neighborhoods $N_x=(x,T\setminus\{x\})$ with one node $x\in T$ in the center and the $k-1$ other nodes in $T\setminus \{x\}$ as $1$-hop neighbors of $x$. This is a valid partial neighborhood as $|T\setminus \{x\}|\leq k-2\leq \Delta$. We call a color $c\in[q]$ a \emph{candidate color} for $N_x$ if there exists a set of input colors $B$ (of size $\leq \Delta)$ satisfying $T\setminus \{x\}\subseteq B\subseteq [m]$ such that $\mathcal{A}\big((x,B)\big)=c$, i.e., algorithm $\mathcal{A}$ outputs color $c$ on the neighborhood $(x,B)$. 
	
	\begin{claim} \label{claim:disjointCandidates}
		For $x\neq x'\in T$ the sets of candidate colors of $N_x$ and $N_{x'}$ are disjoint. 
		\end{claim}
	\begin{proof}
		Assume for contradiction, that $c\in [q]$ is a candidate color for $N_x$ and $N_{x'}$ and let $B\supseteq T\setminus\{x\}$ and $B'\supseteq T\setminus\{x'\}$ be the respective sets such that $c=\mathcal{A}\big((x,B)\big)=\mathcal{A}\big((x',B)\big)$. As $x\in B'$ and $x'\in B$, the neighborhoods $(x,B)$ and $(x',B')$ can occur next to each other in a graph, a contradiction to the correctness of $\mathcal{A}$. 
\renewcommand{\qed}{\ensuremath{\hfill\blacksquare}}
\end{proof}
\renewcommand{\qed}{\hfill \ensuremath{\Box}}
By \Cref{claim:disjointCandidates} and by the pigeonhole principle there exists one $x_*\in T$ for which $N_x$ has at most $\alpha$ candidate colors where
	\begin{align*}
		\alpha=\Bigl\lfloor \frac{q}{|T|}\Bigr\rfloor 
		&	=\Bigl\lfloor \frac{q}{k}\Bigr\rfloor
			=\Bigl\lfloor \frac{m-k}{k}\Bigr\rfloor
			=\Bigl\lfloor \frac{k\Delta-k^2+2k-1}{k}\Bigr\rfloor\\
		&	=\Bigl\lfloor \frac{k(\Delta-k+2)-1}{k}\Bigr\rfloor 
		 =(\Delta-k+2)-1
		 =\Delta-(k-1).
	\end{align*}
Now, let $C_*\subseteq [q]$ be the set of candidate colors of $N_{x_*}$. As $x_*$ is in $T$ and all colors in $T$ are sensitive, for each $c\in C_*$ there exists some $\phi_{c}\in [m]$ such that $\mathcal{A}$ does not output $c$ for $v$ whenever $\phi_{c}$ is the input color of one of $v$'s neighbors. We conclude that the $1$-hop neighborhood $\tilde{N}_{x_*}=(x_*,\{\phi_c\mid c\in C_*\}\cup (T\setminus\{x\}))$ cannot be colored by $\mathcal{A}$, a contradiction. 
The choice of parameters is important. The constructed neighborhood $\tilde{N}_{x_*}$ is a feasible neighborhood as 	$|\{\phi_c\mid c\in C_*\}\cup (T\setminus\{x\})|\leq \alpha+k-1=\Delta$.
\end{proof}

\Cref{lem:oneRoundLB} implies a heuristic lower bound of $\Omega(\Delta)$ to reduce a $\Delta^2/2$-coloring to a $\Delta^2/5$-coloring (if you have $\leq \Delta^2/4$ input colors you can remove at most $\Delta/2$ colors per iteration).  In contrast, the algorithm from \Cref{cor:allInOne} (for a suitable choice of $k$) can reduce a $\Delta^4$-coloring to a $\Delta^2/5$ coloring in $O(1)$ rounds. Thus, the iterative application of \emph{tight} bounds for $1$-round algorithms can be beaten significantly by a simple $O(1)$-round algorithm. This suggests that it is important to understand  constant-time algorithms to settle the complexity of  distributed graph coloring problems. 

\section{Conclusion}
\label{sec:conclusion}
In the current paper we have seen a simple algorithm for distributed graph coloring in which  each vertex locally computes a permutation of the output colors and then \emph{tries} them in batches. A a trial is \emph{successful} if there is no \emph{conflict}, that is, no neighbor  tries the same color in the same round and no neighbor is already permanently colored with that color. Depending on the size of the batches, this algorithm scales between Linial's famous color reduction \cite{linial92} and the locally iterative algorithm by Barenboim, Elkin and Goldenberg \cite{BEG18}. If nodes tolerate conflicts up to a certain threshold the same algorithm can be used to obtain the defective coloring algorithms of \cite{Kuhn2009WeakColoring,BarenboimEK14} and \cite{BE09,BarenboimEK14}, as well as obtaining a simpler algorithm (as compared to \cite{BEG18,Barenboim16}) to compute low out degree colorings aka \emph{arbdefective colorings}. 
The latter are one of the two crucial ingredients in the state of the art $(\Delta+1)$-coloring algorithm in \cite{MT20}. The second ingredient is a $2$-round \emph{list version} of Linial's color reduction, together with the observation that the degree bound can be replaced with a bound on the outdegree. One can also see our algorithm as an extension of Linial's algorithm, or the other way around: In the setting where nodes can only try one color per round ($k=1$) one wants to get colored with one out of $O(\Delta)$ output colors; this process is guaranteed to be successful in $O(\Delta)$ rounds if vertices try colors in a suitable order. Now, if you want to try more than one color per iteration, that is, you want to compress several rounds of the original algorithm into one iteration,  you need to also mark each trial with the round number in which you would have tried it in the original algorithm, yielding an $O(\Delta^2)$-coloring if you want to execute all $O(\Delta)$ rounds in one iteration. We find it astonishing, that in hindsight many crucial results in this area can be related to the algorithm that was presented in Linial's seminal paper \cite{linial92} roughly 30 years ago.

 On the lower bound side his initial $\Omega(\logstar n)$ bound is still the state of the art. The only progress is in terms of understanding $1$-round algorithms or weak variants of the \LOCAL model \cite{disc16_coloring}. Our paper showed that there is a large discrepancy between iterating the best $1$-round algorithm for $O(1)$ times and what can be achieved by a \emph{'smart'} algorithm that uses $O(1)$ rounds. This suggests that we first need to understand constant-time algorithms, both from the upper and lower bound side, before we can settle the complexity of the $(\Delta+1)$-coloring problem. Another approach would be to attack the coloring problem through lower bounds for ruling sets, as there is recent progress for the latter \cite{BBO20,BSKO21}. A large lower bound for a $(2,r)$-ruling set would imply a lower bound for graph coloring via \Cref{lem:rulingSet}; however, it is unclear whether large enough lower bounds for ruling sets exist.  

We end with an additional observation. We purposely keep the observation informal as we merely include it as an additional intuitive guide for the search of the \emph{right} lower bound questions.  We believe that a formal statement would actually hinder the creativity in this process.
\begin{observation}[informal]
	Modulo a $\log \Delta$-factor the difficult part of the $(\Delta+1)$-coloring problem is to reduce a $(1+\eps)\Delta$ coloring  to a $(\Delta+1)$-coloring. 
\end{observation}
\begin{proof}[Proof sketch]
	Assume an algorithm $\mathcal{A}$ that reduces the number of colors from $(1 + \eps)\Delta$ to
	$\Delta + 1$. Now, assume an input coloring with $m\gg(1+\eps\Delta)$ colors is given. Then, one can chop $[m]$ into $x\approx m/((1+\eps)(\Delta+1))$ disjoint color spaces, each of size $(1+\eps)\Delta$ and run $\mathcal{A}$ on each of them in parallel, using a disjoint output color space for each application.  This uses $x\cdot (\Delta+1)\approx m/(1+\eps)$ output colors, i.e., we have reduced the number of colors by a  constant factor (if $\eps$ is constant). Thus, if we begin with $m=O(\Delta^2)$ colors, we obtain a $(\Delta+1)$-coloring with a $O(\log_{1+\eps} \Delta)$ multiplicative overhead. 
\end{proof}
If $\Delta$ is a large enough $\poly\log n$, then randomized algorithm can very efficiently compute   $(1+\eps)\Delta$-colorings \cite{SW10,CLP20,HN21}.

\section*{Acknowledgments}
This project was partially supported by the European Union's Horizon 2020 Research and  Innovation Programme under grant agreement no. 755839. We thank the several people with whom we have discussed the presented algorithm, in particular, Fabian Kuhn and Janosch Deurer. 

\bibliographystyle{alpha}
\bibliography{references}

\end{document}